\documentclass[conference]{IEEEtran}
\usepackage{amsmath,amsfonts,amssymb,bm}
\usepackage{graphicx,subfigure}
\usepackage{algorithm,algpseudocode}
\usepackage{amsthm}     
\usepackage{cite}       
\usepackage{url}

\allowdisplaybreaks[3]  
\newtheorem{thm}{Theorem}

\newcommand{\Nset}{\ensuremath{\mathcal{N}}}
\IEEEoverridecommandlockouts 
\title{Optimal Power Allocation for Two-Way Decode-and-Forward OFDM Relay Networks}
\author{\IEEEauthorblockN{Fei He$^*$, Yin Sun$^*$, Xiang Chen, Limin Xiao, Shidong Zhou}
\IEEEauthorblockA{State Key Laboratory on Microwave and Digital Communications\\
Tsinghua National Laboratory for Information Science and Technology\\
Department of Electronic Engineering, Tsinghua University, Beijing 100084, China\\
Email: hef08@mails.tsinghua.edu.cn, sunyin02@gmail.com}
\thanks{$^*$Fei He and Yin Sun contribute equally to this work. This work is supported by National Basic Research Program of China (2012CB316002), National S\&T Major Project (2010ZX03005-003), National NSF of China (60832008), China's 863 Project (2009AA011501), Tsinghua Research Funding (2010THZ02-3), PCSIRT, International Science Technology Cooperation Program (2010DFB10410) and Tsinghua-Qualcomm Joint Research Program.}}

\begin{document}
\maketitle

\begin{abstract}
This paper presents a novel two-way decode-and-forward (DF) relay strategy for Orthogonal Frequency Division Multiplexing (OFDM) relay networks. This DF relay strategy employs multi-subcarrier joint channel coding to leverage frequency selective fading, and thus can achieve a higher data rate than the conventional per-subcarrier DF relay strategies. We further propose a low-complexity, optimal power allocation strategy to maximize the data rate of the proposed relay strategy. Simulation results suggest that our strategy obtains a substantial gain over the per-subcarrier DF relay strategies, and also outperforms the amplify-and-forward (AF) relay strategy in a wide signal-to-noise-ratio (SNR) region.
\end{abstract}


\section{Introduction}
In recent years, relaying has emerged as a powerful technique to improve the coverage and throughput of wireless networks. Compared with the traditional one-way relaying, two-way relaying provides better spectral efficiency, where two terminal nodes employs an intermediate  relay node to exchange information simultaneously \cite{Rankov_ISIT06,Kim_TIT08}.

Orthogonal Frequency Division Multiplexing (OFDM) is an essential broadband transmission technique to improve the spectral efficiency of wireless networks. A combination of OFDM and relaying techniques has been advocated by many industry standardization groups of next generation wireless networks, such as IEEE 802.16m and 3GPP's LET-Advanced.

In one-way OFDM relay networks, multi-subcarrier joint decode-and-forward (DF) relaying was studied in \cite{Liang_TIT07,Hsu_TSP11,Sun_TSP11}, which can achieve higher data rate than per-subcarrier DF relaying. For two-way OFDM relay networks, the amplify-and-forward (AF) relay strategies were commonly adopted \cite{Ho_ICC08,Kang_VTC09F,Dong_ICASSP10}. However, their performance is quite poor in the low signal-to-noise-ratio (SNR) region due to the amplified noises. The \emph{per-subcarrier} DF relay strategies were considered in \cite{Jitvan_TVT09,Ho_ICC10,Li11}, which are essentially simple accumulations of narrow-band two-way DF relaying over the individual subcarriers. Unfortunately, these strategies suffer from rate losses due to channel mismatching.

In this paper, we propose a novel \emph{multi-subcarrier DF} relay strategy for two-way OFDM relay networks. By performing channel coding across subcarriers, this strategy can exploit frequency selective fading, and achieve higher data rate than the per-subcarrier DF relay strategy in \cite{Jitvan_TVT09}. We further formulate a power allocation problem to maximize the \emph{exchange rate}, which is defined as the maximal data rate can be simultaneously achieved in both directions. An efficient dual decomposition algorithm is proposed to resolve this problem, which has a linear complexity with respect to the number of subcarriers. Simulation results show that the proposed multi-subcarrier DF relay strategy outperforms not only the conventional per-subcarrier DF relay strategy, but also the AF relay strategy in a wide SNR region.

\section{System Description} \label{sec:system}
We consider the two-way OFDM relay network shown in Fig.~\ref{fig:two-way-OFDM}: two terminal nodes $T_1$ and $T_2$ exchange information via an intermediate relay node $T_R$. Assume that each node has a single antenna and operates in a half-duplex mode, i.e., transmitting and receiving in orthogonal time slots \cite{Rankov_ISIT06,Kim_TIT08}. All the nodes employ OFDM air interface with the same $N$ subcarriers.

\begin{figure}[h]
    \centering
    \subfigure[The MA phase.]{
    \scalebox{0.85}{\includegraphics*[68,676][313,772]{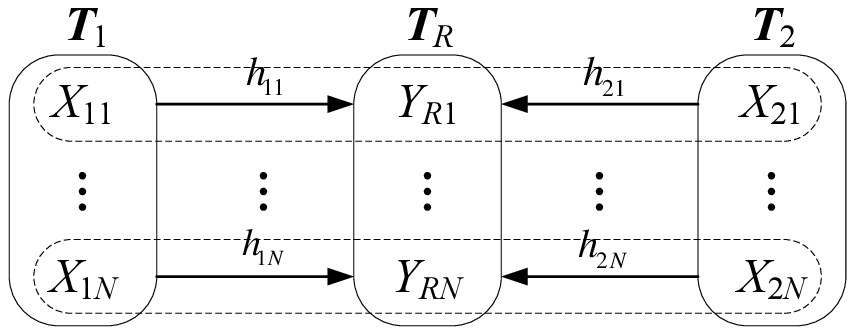}} \label{fig:MA-phase}}
    \subfigure[The BC phase.]{
    \scalebox{0.85}{\includegraphics*[68,557][313,651]{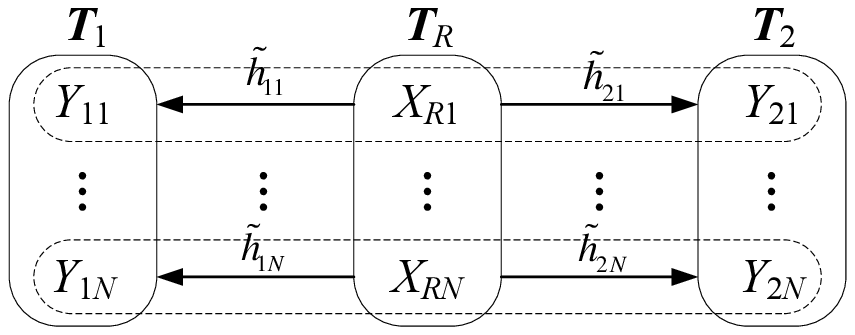}} \label{fig:BC-phase}}
    \caption{System model of two-way OFDM relay network.} \label{fig:two-way-OFDM}
\end{figure}

The DF relay procedure comprises of a multiple-access (MA) phase and a broadcast (BC) phase without direct transmissions, as shown in Fig.~\ref{fig:two-way-OFDM}. We set $X_i\!=\!(X_{i1},\ldots,X_{iN})$ and $Y_i\!=\!(Y_{i1},\ldots,Y_{iN})$ for $i\!=\!1,2,R$, where $X_{in}$ and $Y_{in}$ denote the normalized transmitted and received signal in the $n$th subcarrier at $T_i$. In the MA phase, $T_1$ and $T_2$ transmit $X_1$ and $X_2$ simultaneously to the relay node, and the relay $T_R$ performs multi-user detection to fully decode ${X}_1$ and ${X}_2$ from the received $Y_R$; in the BC phase, the relay $T_R$ broadcasts $X_R=f(\hat{X}_1,\hat{X}_2)$ to $T_1$ and $T_2$.

Specifically, in the $n$th subcarrier, $h_{1n}$ and $h_{2n}$ denote the channel coefficients from $T_1$ and $T_2$ to $T_R$, respectively, $\tilde{h}_{1n}$ and $\tilde{h}_{2n}$ denote the channel coefficients from $T_R$ to $T_1$ and $T_2$, respectively. Thus, the received signals $Y_{in}$'s in the $n$th subcarrier at $T_i$'s are given by
{\setlength\arraycolsep{2pt} \begin{eqnarray}
Y_{Rn}&=&\sqrt{P_{1n}}h_{1n}X_{1n}+\sqrt{P_{2n}}h_{2n}X_{2n}+Z_{Rn}, \label{eq:mac-channel}\\
Y_{1n}&=&\sqrt{P_{Rn}}\tilde{h}_{1n}X_{Rn}+Z_{1n}, \\
Y_{2n}&=&\sqrt{P_{Rn}}\tilde{h}_{2n}X_{Rn}+Z_{2n}, \label{eq:bc-channel}
\end{eqnarray}}
where $P_{in}$ denotes the transmit power in the $n$th subcarrier at $T_i$, and $Z_{in}$ denotes independent complex \emph{additive white Gaussian noises} with zero mean and unit variance, i.e., $Z_{in}\!\sim\!\mathcal{CN}(0,1)$, for $i\!=\!1,2,R$. Therefore, $P_{in}$ essentially denotes the corresponding transmit SNR.

We assume that ${\mu}\!\in\!(0,1)$ denotes the \emph{fixed} proportion of time slot allocated to the MA phase, and all the terminals are subject to separate power constraints $\sum_{n=1}^N\!{P_{in}}\leq P_{i\!\max}$ ($i\!=\!1,2,R$), where $P_{i\!\max}$ denotes the maximum available power for $T_i$.

\section{A novel DF relay strategy for two-way \\OFDM relay networks} \label{sec:rate-region}
For two-way OFDM relay networks, the conventional DF relay strategies simply applied the narrow-band DF technique over each subcarrier independently, and the overall throughput was the sum rate of all the subcarriers \cite{Jitvan_TVT09,Li11}. However, these strategies suffer from rate losses due to channel mismatching.

In this section, we propose a novel \emph{multi-subcarrier} DF relay strategy, which performs channel coding across subcarriers to leverage frequency selective fading, and achieve higher data rate. An achievable rate region for this strategy is provided in the following theorem.
\begin{thm} \label{thm1}
Any rate pair $(R_{12},R_{21})$ satisfying the following inequalities is achievable for the two-way OFDM relay network given by \eqref{eq:mac-channel}-\eqref{eq:bc-channel}:
\begin{subequations} \label{eq:thm1}
\begin{align}
R_{12}&\leq{\mu}\sum_{n=1}^N{\log_2\!\left(1+|h_{1n}|^{2}P_{1n}\right)},\\[-3pt]
R_{12}&\leq(1-{\mu})\sum_{n=1}^N{\log_2\!\left(1+|\tilde{h}_{2n}|^{2}P_{Rn}\right)},\\[-3pt]
R_{21}&\leq{\mu}\sum_{n=1}^N{\log_2\!\left(1+|h_{2n}|^{2}P_{2n}\right)},\\[-3pt]
R_{21}&\leq(1-{\mu})\sum_{n=1}^N{\log_2\!\left(1+|\tilde{h}_{1n}|^{2}P_{Rn}\right)},\\[-3pt]
R_{12}+R_{21}&\leq{\mu}\sum_{n=1}^N{\log_2\!\left(1+|h_{1n}|^{2}P_{1n}+|h_{2n}|^{2}P_{2n}\right)},\\[-3pt]
\sum_{n=1}^N{P_{in}}&\leq P_{i\!\max},\;i=1,2,R, \\
P_{in}&\geq0,\;\forall~n\in\Nset,\;i=1,2,R.
\end{align}
\end{subequations}
where $\Nset\triangleq\{1,\ldots,N\}$, $R_{12}$ and $R_{21}$ denote the achievable data rates from $T_1$ to $T_2$ and from $T_2$ to $T_1$, respectively.
\end{thm}
\begin{proof}
An achievable rate region of a DF relay strategy for the discrete-memoryless two-way relay network has been given by the set of $(R_{12},R_{21})$ rate pairs satisfying \cite{Knopp06}
\begin{subequations}
\begin{align}
R_{12}&\leq\min\!\left\{{\mu}I(X_1;Y_R|X_2),(1-{\mu})I(X_R;Y_2)\right\}, \label{eq:ineq1}\\
R_{21}&\leq\min\!\left\{{\mu}I(X_2;Y_R|X_1),(1-{\mu})I(X_R;Y_1)\right\}, \label{eq:ineq2}\\
\!\!\!\!R_{12}+R_{21}&\leq{\mu}I(X_1,X_2;Y_R). \label{eq:ineq3}
\end{align}
\end{subequations}
Each mutual information item in \eqref{eq:ineq1}-\eqref{eq:ineq3} corresponds to the achievable rate of a parallel point-to-point channel. Similar to the idea in the proof of \cite[Theorem 1]{Liang_TIT07}, we choose the input signals for each subcarrier to be \emph{independent Gaussian} distributed with unit variance, i.e., $X_{in}\!\sim\!\mathcal{CN}(0,1)$. Thus, each mutual information item in \eqref{eq:ineq1}-\eqref{eq:ineq3} is replaced by the sum of $N$ logarithmic rate items decided by \eqref{eq:mac-channel}-\eqref{eq:bc-channel} with separate power constraints. The proof is complete.
\end{proof}

\emph{Remark 1:} The key idea of this \emph{multi-subcarrier} two-way DF relay strategy is introducing channel coding across subcarriers to fully exploit frequency selective fading. Implicitly, the information transmitted over one subcarrier in the MA phase may be forwarded over some other subcarriers in the BC phase. By this, the problem from \emph{mismatching} of wireless channels over subcarriers is resolved. The achievable rate region of Theorem~\ref{thm1} is no smaller than that achieved by the \emph{per-subcarrier} two-way DF relaying, which is the set of rate pairs satisfying \cite{Jitvan_TVT09}
\begin{subequations}
\begin{align}
R_{12}\leq\sum_{n=1}^N&{\min\!\left\{{\mu}\log_2\!\left(1+|h_{1n}|^{2}P_{1n}\right),\right.}\nonumber\\[-7pt]
&{\big.(1-{\mu})\log_2\!\big(1+|\tilde{h}_{2n}|^{2}P_{Rn}\big)\big\}}, \nonumber\\[-3pt]
R_{21}\leq\sum_{n=1}^N&{\min\!\left\{{\mu}\log_2\!\left(1+|h_{2n}|^{2}P_{2n}\right),\right.}\nonumber\\[-7pt]
&{\big.(1-{\mu})\log_2\!\big(1+|\tilde{h}_{1n}|^{2}P_{Rn}\big)\big\}}, \nonumber\\[-3pt]
R_{12}+R_{21}\leq{\mu}\sum_{n=1}^N&{\log_2\!\left(1+|h_{1n}|^{2}P_{1n}+|h_{2n}|^{2}P_{2n}\right)}.\nonumber
\end{align}
\end{subequations}

Therefore, multi-subcarrier two-way relay channel is not a simple linear combination of multiple narrow-band single-subcarrier two-way relay subchannels. Similar observations have been found for \emph{one-way} parallel relay networks \cite{Liang_TIT07}.

\section{Optimal Power Allocation} \label{sec:power-allocation}
In this section, we investigate the largest achievable symmetric exchange data rate of our proposed DF relay stratey, which can be approached by an optimal power allocation.
\subsection{Problem Formulation}
By optimizing the power allocation strategy $(\bm{P}_1,\bm{P}_2,\bm{P}_R)$, our objective is to maximize the \emph{exchange rate} $R_\mathrm{X}=\min\left\{R_{12},\,R_{21}\right\}$, which is defined as the data rate can be achieved simultaneously in both directions, where $\bm{P}_i=[P_{i1},P_{i2},\ldots,P_{iN}]^T$ denotes the power allocation vector at $T_i$, for $i\!=\!1,2,R$. This can be expressed as the following \emph{convex} optimization problem:
\begin{subequations} \label{eq:original-problem}
\begin{align}
\!\!\!\!\max_{\bm{P}_1,\bm{P}_2,\bm{P}_R,R_\mathrm{X}}&R_\mathrm{X} \\[-6pt]
\mathrm{s.t.}~R_\mathrm{X}&\leq{\mu}\sum_{n=1}^N{\log_2\!\left(1+|h_{1n}|^{2}P_{1n}\right)},\label{eq:MA-constr1}\\[-3pt]
R_\mathrm{X}&\leq(1-{\mu})\sum_{n=1}^N{\log_2\!\left(1+|\tilde{h}_{2n}|^{2}P_{Rn}\right)},\label{eq:BC-constr1}\\[-3pt]
R_\mathrm{X}&\leq{\mu}\sum_{n=1}^N{\log_2\!\left(1+|h_{2n}|^{2}P_{2n}\right)},\label{eq:MA-constr2}\\[-3pt]
R_\mathrm{X}&\leq(1-{\mu})\sum_{n=1}^N{\log_2\!\left(1+|\tilde{h}_{1n}|^{2}P_{Rn}\right)},\label{eq:BC-constr2}\\[-3pt]
R_\mathrm{X}&\leq\frac{\mu}{2}\sum_{n=1}^N{\log_2\!\left(1+|h_{1n}|^{2}P_{1n}+|h_{2n}|^{2}P_{2n}\right)},\label{eq:MA-constr3}\\[-3pt]
\sum_{n=1}^N&{P_{in}}\leq P_{i\!\max},\;i=1,2,R, \\
P_{in}&\geq0,\;\forall~n\in\Nset,\;i=1,2,R.
\end{align}
\end{subequations}

It is readily observed that in the problem \eqref{eq:original-problem}, $\bm{P}_1$ and $\bm{P}_2$ are only related to the constraints \eqref{eq:MA-constr1} \eqref{eq:MA-constr2} \eqref{eq:MA-constr3}, while $\bm{P}_R$ is only related to the constraints \eqref{eq:BC-constr1} \eqref{eq:BC-constr2}. This observation helps to decompose our original power allocation problem \eqref{eq:original-problem} into the following two subproblems:
{\setlength\arraycolsep{2pt} \begin{eqnarray}
\!\!\!\!\max_{\bm{P}_1,\bm{P}_2,R_\mathrm{MA}}&&R_\mathrm{MA} \label{eq:MA-problem}\\[-6pt]
\mathrm{s.t.}&&R_\mathrm{MA}\leq{\mu}\sum_{n=1}^N{\log_2\!\left(1+|h_{1n}|^{2}P_{1n}\right)},\nonumber\\[-3pt]
&&R_\mathrm{MA}\leq{\mu}\sum_{n=1}^N{\log_2\!\left(1+|h_{2n}|^{2}P_{2n}\right)},\nonumber\\[-3pt]
&&R_\mathrm{MA}\leq\frac{\mu}{2}\sum_{n=1}^N{\log_2\!\left(1+|h_{1n}|^{2}P_{1n}+|h_{2n}|^{2}P_{2n}\right)},\nonumber\\[-3pt]
&&\sum_{n=1}^N{P_{1n}}\leq P_{1\!\max},\;\sum_{n=1}^N{P_{2n}}\leq P_{2\!\max},\nonumber\\
&&P_{1n}\geq0,P_{2n}\geq0,\;\forall~n\in\Nset. \nonumber\\[6pt]
\max_{\bm{P}_R,R_\mathrm{BC}}&&R_\mathrm{BC} \label{eq:BC-problem}\\[-6pt]
\mathrm{s.t.}&&R_\mathrm{BC}\leq(1-{\mu})\sum_{n=1}^N{\log_2\!\left(1+|\tilde{h}_{1n}|^{2}P_{Rn}\right)},\nonumber\\[-3pt]
&&R_\mathrm{BC}\leq(1-{\mu})\sum_{n=1}^N{\log_2\!\left(1+|\tilde{h}_{2n}|^{2}P_{Rn}\right)},\nonumber\\[-3pt]
&&\sum_{n=1}^N{P_{Rn}}\leq P_{R\!\max},\;P_{Rn}\geq0,\;\forall~n\in\Nset. \nonumber
\end{eqnarray}}

We can denote $R_\mathrm{MA}^{\star}$ and $R_\mathrm{BC}^{\star}$ as the optimal values for the MA subproblem \eqref{eq:MA-problem} and the BC subproblem \eqref{eq:BC-problem}, respectively. Eventually, the maximal practical exchange rate for our proposed DF strategy is given by $R_\mathrm{X}^{\star}=\min\!\left\{R_\mathrm{MA}^{\star},R_\mathrm{BC}^{\star}\right\}$.

\subsection{Proposed Dual Decomposition Algorithm}
The interior-point methods can be used to solve both of the \emph{convex} optimization problems \eqref{eq:MA-problem} and \eqref{eq:BC-problem}, however, they quickly become computationally intractable as $N$ increases, because they have a complexity of $O(N^3)$ at least when solving the search direction in each iteration \cite{Boyd04}. Therefore, we present a low-complexity dual decomposition algorithm for the subproblems \eqref{eq:MA-problem} and \eqref{eq:BC-problem}, to efficiently obtain the optimal solution to \eqref{eq:original-problem}. Next, we will take the subproblem \eqref{eq:MA-problem} as an example to illustrate this algorithm.

Note that problem \eqref{eq:MA-problem} is strictly feasible. Then, according to the Slater's condition \cite{Boyd04}, it is equivalent with the following dual optimization problem:
\begin{equation} \label{eq:dual-problem}
\max_{\bm{\lambda},\bm{\alpha}\succeq0}\left\{
\min_{\bm{P}_1,\bm{P}_2\succeq0,R_\mathrm{MA}} \mathcal{L}\left(\bm{P}_1,\bm{P}_2,R_\mathrm{MA},\bm{\lambda},\bm{\alpha}\right)\right\},
\end{equation}
where
{\setlength\arraycolsep{0pt} \begin{eqnarray} \label{eq:Lagrangian}
\!\!\!\!&&\mathcal{L}\left(\bm{P}_1,\bm{P}_2,R_\mathrm{MA},\bm{\lambda},\bm{\alpha}\right)=-R_\mathrm{MA}\nonumber\\
&&~~+\,\lambda_1\!\!\left[R_\mathrm{MA}-{\mu}\sum_{n=1}^N{\log_2\!\left(1+|h_{1n}|^{2}P_{1n}\right)}\right]\nonumber\\[-2pt]
&&~~+\,\lambda_2\!\!\left[R_\mathrm{MA}-{\mu}\sum_{n=1}^N{\log_2\!\left(1+|h_{2n}|^{2}P_{2n}\right)}\right]\nonumber\\[-2pt]
&&~~+\,\lambda_3\!\!\left[R_\mathrm{MA}-\frac{\mu}{2}\sum_{n=1}^N{\log_2\!\left(1+|h_{1n}|^{2}P_{1n}+|h_{2n}|^{2}P_{2n}\right)}\right]\nonumber\\[-2pt]
&&~~+\,\alpha_1\!\!\left(\sum_{n=1}^N{P_{1n}}-P_{1\!\max}\right)\!\!+\alpha_2\!\!\left(\sum_{n=1}^N{P_{2n}}-P_{2\!\max}\right)\nonumber\\[-2pt]
&&=\sum_{n=1}^N\bigg[\alpha_1P_{1n}+\alpha_2P_{2n}-\mu\lambda_1\log_2\!\left(1+|h_{1n}|^{2}P_{1n}\right)\bigg.\nonumber\\
&&~~~-\mu\lambda_2\log_2\!\left(1+|h_{2n}|^{2}P_{2n}\right)\nonumber\\
&&~~~-\!\left.\frac{\mu\lambda_3}{2}\log_2\!\left(1+|h_{1n}|^{2}P_{1n}+|h_{2n}|^{2}P_{2n}\right)\right]\nonumber\\
&&~~+\left(\lambda_1+\lambda_2+\lambda_3-1\right)R_\mathrm{MA}-\alpha_1P_{1\!\max}-\alpha_2P_{2\!\max}
\end{eqnarray}}
is the partial Lagrangian of \eqref{eq:MA-problem}, and $\bm{\lambda}=[\lambda_1,\lambda_2,\lambda_3]^T, \bm{\alpha}=[\alpha_1,\alpha_2]^T$ are nonnegative dual variables associated with the three rate constraints and two power constraints, respectively.

According to \eqref{eq:Lagrangian}, the \emph{inner minimization} problem of \eqref{eq:dual-problem} can be decomposed as $N$ independent per-subcarrier power allocation problems. Hence, the computational complexity for solving the inner problem is only linear with respect to $N$. In addition, the optimal $(P_{1n},P_{2n})$ must satisfy the following Karush-Kuhn-Tucker (KKT) conditions for given dual variables $(\bm{\lambda},\bm{\alpha})$ \cite{Boyd04}:
{\setlength\arraycolsep{2pt} \begin{eqnarray}
\frac{\partial{\mathcal{L}}}{\partial{P_{1n}}}&=&
\alpha_1\!-\!\frac{\mu\lambda_3|h_{1n}|^2}{2\ln2(1+|h_{1n}|^2P_{1n}+|h_{2n}|^2P_{2n})}\nonumber\\
&&-\frac{\mu\lambda_1|h_{1n}|^2}{\ln2(1+|h_{1n}|^2P_{1n})}
\left\{\begin{array}{ll}
\geq0&\;\mathrm{if}\;P_{1n}=0 \\[2pt]
=0&\;\mathrm{if}\;P_{1n}>0
\end{array}\right., \label{eq:KKT1}\\[5pt]
\frac{\partial{\mathcal{L}}}{\partial{P_{2n}}}&=&
\alpha_2\!-\!\frac{\mu\lambda_3|h_{2n}|^2}{2\ln2(1+|h_{1n}|^2P_{1n}+|h_{2n}|^2P_{2n})}\nonumber\\
&&-\frac{\mu\lambda_2|h_{2n}|^2}{\ln2(1+|h_{2n}|^2P_{2n})}
\left\{\begin{array}{ll}
\geq0&\;\mathrm{if}\;P_{2n}=0 \\[2pt]
=0&\;\mathrm{if}\;P_{2n}>0
\end{array}\right.. \label{eq:KKT2}
\end{eqnarray}}

Thus, it must belong to one of the following four cases:

\textbf{Case 1:} $P_{1n}>0,P_{2n}>0$. Then the formulas \eqref{eq:KKT1} and \eqref{eq:KKT2} hold with equality. It is hard to solve \eqref{eq:KKT1} and \eqref{eq:KKT2} directly since they are both \emph{quadratic} equations of two variables $P_{1n}$ and $P_{2n}$. However, we can utilize an auxiliary variable defined as $x=|h_{1n}|^2P_{1n}+|h_{2n}|^2P_{2n}$ to simplify them. More specifically, from \eqref{eq:KKT1} and \eqref{eq:KKT2}, one can obtain that
{\setlength\arraycolsep{2pt}\begin{eqnarray}
|h_{1n}|^2P_{1n}&=&\frac{2\mu\lambda_1|h_{1n}|^2}{2\ln2\cdot\alpha_1\!-\!\mu\lambda_3|h_{1n}|^2/(1+x)}-1,\label{eq:power1}\\
|h_{2n}|^2P_{2n}&=&\frac{2\mu\lambda_2|h_{2n}|^2}{2\ln2\cdot\alpha_2\!-\!\mu\lambda_3|h_{2n}|^2/(1+x)}-1.\label{eq:power2}
\end{eqnarray}}

Taking the sum of the above two equations, we obtain a \emph{cubic} equation of $x$, which has closed-form solutions given by \emph{Cardano's Formula} \cite{Dunham90}. After deriving the positive root $x$ of this \emph{cubic} equation, we can easily obtain the optimal $P_{1n}$ and $P_{2n}$ from \eqref{eq:power1} and \eqref{eq:power2}. By this procedure, the quadratic equations \eqref{eq:KKT1} and \eqref{eq:KKT2} are solved analytically by converting to an equivalent \emph{cubic} equation. Finally, we need to check whether $P_{1n}$ and $P_{2n}$ satisfy the conditions $P_{1n}>0,P_{2n}>0$.

\textbf{Case 2:} $P_{1n}>0,P_{2n}=0$. Then the solutions to \eqref{eq:KKT1} and \eqref{eq:KKT2} can be derived as
{\setlength\arraycolsep{2pt}\begin{eqnarray}
P_{1n}&=&\frac{\mu(2\lambda_1+\lambda_3)}{2\ln2\cdot\alpha_1}-\!\frac{1}{|h_{1n}|^2},\\
P_{2n}&=&0.
\end{eqnarray}}
This case happens only if $P_{1n}>0$ and the KKT condition \eqref{eq:KKT2}, $2\ln2{\cdot}\alpha_2\geq2\mu\lambda_2|h_{2n}|^2+\frac{\mu\lambda_3|h_{2n}|^2}{1+|h_{1n}|^2P_{1n}}$, is satisfied.

\textbf{Case 3:} $P_{1n}=0,P_{2n}>0$. Then the KKT conditions can be reformulated as
{\setlength\arraycolsep{2pt}\begin{eqnarray}
P_{1n}&=&0,\\
P_{2n}&=&\frac{\mu(2\lambda_2+\lambda_3)}{2\ln2\cdot\alpha_2}-\!\frac{1}{|h_{2n}|^2}.
\end{eqnarray}}
This case happens only if $P_{2n}>0$ and the KKT condition \eqref{eq:KKT1}, $2\ln2{\cdot}\alpha_1\geq2\mu\lambda_1|h_{1n}|^2+\frac{\mu\lambda_3|h_{1n}|^2}{1+|h_{2n}|^2P_{2n}}$, is satisfied.

\textbf{Case 4:} $P_{1n}=0,P_{2n}=0$. This is the default case when the above three cases do not happen.

Then, we optimize the dual variables $(\bm{\lambda},\bm{\alpha})$ for the \emph{outer maximization} problem of \eqref{eq:dual-problem}. We redefine $\bm{\nu}=[\lambda_1,\lambda_2,\lambda_3,\alpha_1,\alpha_2]^T$. Further, considering the KKT condition for the optimal data rate $R_\mathrm{MA}$, we have
\begin{equation}
\frac{\partial{\mathcal{L}}}{\partial{R_\mathrm{MA}}}=\lambda_1+\lambda_2+\lambda_3-1=0.\label{eq:KKT3}
\end{equation}

In view of that the objective function is not differentiable with respect to $(\bm{\lambda},\bm{\alpha})$, we consider to update $\bm{\nu}$ using the subgradient method \cite{Bertsekas99,Boyd07}. Specifically, in the $k$th iteration, the subgradient method updates $\bm{\nu}^k$ by
\begin{equation} \label{eq:dual-update}
\bm{\nu}^{k+1}=\left[\bm{\nu}^{k}+s^k\bm{\eta}(\bm{\nu}^k)\right]_{\mathcal{P}},
\end{equation}
where $[\bm{\nu}]_{\mathcal{P}}$ represents the \emph{orthogonal projection} of $\bm{\nu}$ to the dual feasible set $\{\bm{\nu}\mid\bm{1}^T\bm{\lambda}=1, \bm{\lambda},\bm{\alpha}\succeq0\}$ based on a finite algorithm in \cite{Michelot86}, $s^k$ is the step size of the $k$th iteration, and $\bm{\eta}(\bm{\nu}^k)$ is the subgradient of the outer problem of \eqref{eq:dual-problem} at $\bm{\nu}^k$, which can be chosen as
{\setlength\arraycolsep{0.5pt}\begin{eqnarray} \label{eq:subgradient}
\!\!\!\!\!\!\!\!\bm{\eta}(\bm{\nu}^k)&=&\!\left[\begin{array}{c}
-{\mu}\sum_{n=1}^N{\log_2\!\left(1+|h_{1n}|^{2}P_{1n}^{\star}\right)} \\[5pt]
-{\mu}\sum_{n=1}^N{\log_2\!\left(1+|h_{2n}|^{2}P_{2n}^{\star}\right)} \\[5pt]
-\frac{\mu}{2}\sum_{n=1}^N{\log_2\!\left(1+|h_{1n}|^{2}P_{1n}^{\star}+|h_{2n}|^{2}P_{2n}^{\star}\right)}\\[5pt]
\sum_{n=1}^N{P_{1n}^{\star}}-P_{1\!\max} \\[5pt]
\sum_{n=1}^N{P_{2n}^{\star}}-P_{2\!\max}
\end{array}\right]\!\!,
\end{eqnarray}}
where $P_{1n}^{\star}$ and $P_{2n}^{\star}$ are the optimal solution of the inner minimization problem in the $k$th iteration. It has been shown that the subgradient updates in \eqref{eq:dual-update} can converge to the optimal dual point $\bm{\nu}^{\star}$ as $k\rightarrow\infty$, provided that the step size $s^k$ is chosen according to a diminishing step size rule \cite{Boyd07}.

Let $C_{\mathrm{MA},i}(\bm{P}_1,\bm{P}_2)(i=1,2,3)$ denote the the right-hand sides of three rate constraints in \eqref{eq:MA-problem}, respectively, and thus we obtain the optimal $R_\mathrm{MA}^{\star}=\min\{C_{\mathrm{MA},i}(\bm{P}_1^{\star},\bm{P}_2^{\star}),\;i\!=\!1,2,3\}$.

The proposed dual decomposition algorithms for the MA subproblem \eqref{eq:MA-problem} are summarized in Algorithm~\ref{alg1}. Similarly, the BC subproblem \eqref{eq:BC-problem} can be solved with the same techniques. Their complexity grow in the order of $O(N)$, which are much lower than the classic convex optimization software package based on interior-point methods. Therefore, our proposed algorithm is more favorable for large value of $N$, which is quite typical in OFDM systems.
\begin{algorithm}
\caption{Proposed dual decomposition algorithm for \eqref{eq:MA-problem}} \label{alg1}
\begin{algorithmic}[1]  
\State \textbf{Input} the system parameters $\!\{N,P_{1\!\max},P_{2\!\max}\}\!$, the channel coefficients $\{h_{1n},h_{2n}\}_{n=1}^N$, and a solution accuracy $\epsilon$.
\State Set $k=1$; Initialize dual variables $\bm{\nu}^1=\bm{1}$.
\Repeat
\State Compute the optimal $\{P_{1n},P_{2n}\}$ according to \eqref{eq:KKT1} and \eqref{eq:KKT2} for $\forall~n\in\Nset$;
\State Update the dual variables $\bm{\nu}^{k}$ according to \eqref{eq:dual-update};
\State $k:=k+1$;
\Until $\|\bm{\nu}^{k}-\bm{\nu}^{k-1}\|\leq\epsilon\,\|\bm{\nu}^{k-1}\|$.
\State \textbf{Output} the optimal primal solution $\{\bm{P}_1^{\star},\bm{P}_2^{\star}\}$ and $R_\mathrm{MA}^{\star}\!=\!\min\{C_{\mathrm{MA},i}(\bm{P}_1^{\star},\bm{P}_2^{\star}),\;i\!=\!1,2,3\}$.
\end{algorithmic}
\end{algorithm}

\section{Simulation Results} \label{sec:simulation}
We consider an OFDM system with $N=32$ subcarriers. The frequency-domain channels are generated using 8 independently and identically distributed Rayleigh distributed time-domain taps with unit variance \cite{Ho_ICC08}. The separate power constraints are set as $P_{1\!\max}=P_{2\!\max}=P_{R\!\max}$, and $\mu=0.5$.

Our proposed multi-subcarrier DF relay strategy is denoted as ``Type~1 DF'' scheme. Two reference schemes are considered in our simulations: The first one is the \emph{per-subcarrier} two-way DF OFDM relay strategy in \cite{Jitvan_TVT09}, which is denoted as ``Type~2 DF'' scheme; the second one is the two-way AF OFDM relaying scheme with optimized tone permutation in \cite{Ho_ICC08}. We divide the sum rate (approximated by the lower bound $2R_\mathrm{X}$ in Type~1/2 DF scheme) by $N$ and use this per-subcarrier sum rate to evaluate performance at different average SNRs, which are only related with the power constraints $P_{i\!\max}$'s.

Fig.~\ref{fig:result} presents the performance of different two-way OFDM relay strategies. The best performance is achieved by Type~1 DF scheme with optimal power allocation (PA). At the spectral efficiency of 2~bits/s/Hz, Type~1 DF scheme with optimal PA provides a coding gain of about 2.5~dB compared with Type~2 DF scheme, by performing channel coding across subcarriers. The PA gain between optimal PA and uniform PA of Type~1 DF scheme is given by 1.6~dB. It is interesting that Type~1 DF scheme with uniform PA even outperforms Type~2 DF scheme with optimal PA, when the average SNR is in the region [0~dB, 20~dB].

Although Type~1 DF scheme has no advantage over the AF scheme in the high SNR region due to its additional fully decoding requirement at $T_R$, it outperforms the AF scheme in the low and median SNR region. The intersection of the curves for Type~1 DF scheme and the AF scheme is at about 17.5~dB, which is 5~dB higher than that for Type~2 DF scheme and the AF scheme.

\section{Conclusion} \label{sec:conclusion}
We have proposed a novel DF relay strategy for two-way OFDM relay networks and derived its achievable rate region. The key idea is making use of cross-subcarrier channel coding to fully exploit frequency selective fading. An efficient duality-based power allocation algorithm is also proposed to maximize the symmetric exchange data rate in both directions. Our simulation results suggest that the proposed DF strategy has better performance than existing DF or AF two-way OFDM relay strategies in the moderately low SNR region. We believe this two-way DF strategy tends to be optimal, i.e., achieving the capacity region outer bound, in the moderately low SNR region. The optimality of the proposed two-way DF strategy and the effect of channel uncertainty are currently under our investigation.

\begin{figure}[t]
    \centering
    \scalebox{0.55}{\includegraphics*[70,422][485,746]{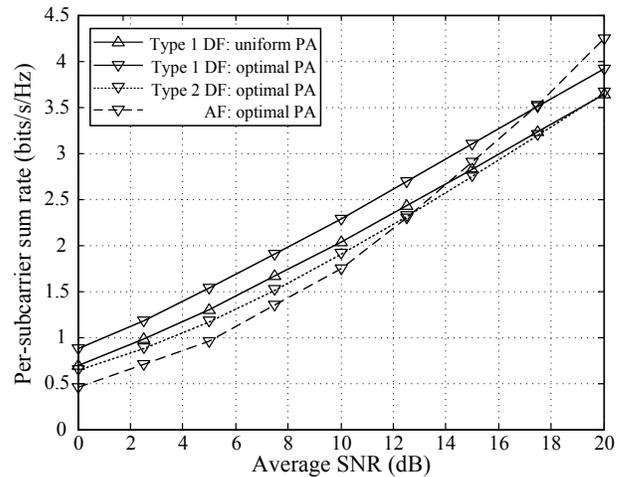}}
    \caption{Per-subcarrier sum rate of different two-way OFDM relaying schemes.} \label{fig:result}
\end{figure}


\bibliographystyle{IEEEtran}
\bibliography{refs_ICC12}

\end{document}